\documentclass[12pt,a4paper]{article}
\usepackage{caption}
\usepackage[utf8]{inputenc}
\usepackage[T1]{fontenc}
\usepackage[british]{babel}
\usepackage[a4paper,left=2.4cm,right=2.4cm,top=2.5cm,bottom=2.5cm]{geometry}
\usepackage{amsmath,amssymb,amsthm,mathtools}
\usepackage{graphicx}
\usepackage{hyperref}
\usepackage[capitalize,noabbrev]{cleveref}
\usepackage{booktabs}
\usepackage[expansion=false]{microtype}
\usepackage{xcolor}
\usepackage{bm}
\usepackage{enumitem}
\usepackage{algorithm}
\usepackage{algpseudocode}
\usepackage{titlesec}
\usepackage[numbers,sort&compress,square]{natbib}
\hypersetup{colorlinks,linkcolor=blue!50!black,citecolor=blue!50!black,
            urlcolor=blue!50!black,breaklinks=true}
\titleformat*{\section}{\large\bfseries}
\titleformat*{\subsection}{\normalsize\bfseries}
\titleformat*{\subsubsection}{\normalsize\itshape}
\newtheorem{theorem}{Theorem}[section]
\newtheorem{proposition}[theorem]{Proposition}
\newtheorem{lemma}[theorem]{Lemma}
\newtheorem{corollary}[theorem]{Corollary}
\theoremstyle{definition}

\newtheorem{assumption}[theorem]{Assumption}
\newtheorem{remark}[theorem]{Remark}

\DeclareMathOperator{\im}{Im}

\DeclareMathOperator{\diag}{diag}

\DeclareMathOperator{\spec}{spec}
\DeclareMathOperator{\cond}{cond}
\newcommand{\R}{\mathbb{R}}
\newcommand{\C}{\mathbb{C}}

\newcommand{\Hil}{\mathcal{H}}
\newcommand{\norm}[1]{\left\lVert #1 \right\rVert}
\newcommand{\inner}[2]{\langle #1, #2 \rangle}
\newcommand{\dd}{\,\mathrm{d}}
\newcommand{\eps}{\varepsilon}
\newcommand{\smin}{\sigma_{\mathrm{min}}}
\newcommand{\smax}{\sigma_{\mathrm{max}}}
\newcommand{\Geps}{\Gamma_{\!\eps}}
\newcommand{\bdf}{\bm{c}}
\newcommand{\Mc}{\mathcal{M}_{\!C}}
\title{\bfseries
  Tikhonov-regularised projected gradient flow for
  equality-constrained bilinear quantum control}
\author{Tanveer Ahmad\textsuperscript{1,\,*}%
  \\[6pt]
  \normalsize\textsuperscript{1}\textit{Hunan Key Laboratory of
  Super-MicroStructure and Ultrafast Process,}\\
  \normalsize\textit{School of Physics, Central South University,
  Changsha 410083, China}\\[4pt]
  \normalsize\textsuperscript{*}\href{mailto:tanveer.quantum@gmail.com}%
  {\texttt{tanveer.quantum@gmail.com}}}
\date{(Dated: \today)}
\begin{document}
\maketitle
\begin{abstract}
\noindent
We study a projection-type gradient flow for the equality-constrained
maximisation of a smooth bilinear control objective on
$\Hil=L^2(0,T;\R)$, in which the Lagrange multipliers are eliminated
through an $(M{+}1)\times(M{+}1)$ moving Gram matrix
$\Gamma(s)_{\ell\ell'}=\int_0^T S(t)\,c_\ell(s,t)\,c_{\ell'}(s,t)\dd t$
assembled from the gradient of the objective and the gradients of the
$M$ scalar constraints. The flow generates monotonic ascent in
continuous-$s$ time but becomes unstable on discretisation, and
existing implementations rely on heuristic step-size safeguards whose
mathematical content has not been fully analysed. We close this gap by
replacing $\Gamma$ with the Tikhonov-regularised matrix
$\Geps:=\Gamma+\eps^{2}I$, $\eps\ge 0$, and we prove the following.
First, an exact spectral identity yields the conditioning formula
$\cond(\Geps)=(\smax^{2}+\eps^{2})/(\smin^{2}+\eps^{2})$. Second,
objective monotonicity $\dd J/\dd s\ge 0$ persists for every
$\eps\ge 0$. Third, an exact constraint-drift identity gives
$|h_{m}[E_{\eps}(s)]-C_{m}|=\mathcal{O}(\eps^{2})$ with a computable
prefactor. Fourth, the regularised trajectory converges to the
unregularised one in $L^{2}(0,T)$ at rate $\mathcal{O}(\eps^{2})$,
under uniform invertibility of $\Gamma$. Fifth, a discrete
Courant--Friedrichs--Lewy criterion
$\Delta s\,G\,\norm{\Geps^{-1}}\le \alpha<2$ guarantees objective
monotonicity of the forward-Euler scheme up to
$\mathcal{O}(\Delta s^{2})$ local truncation error. The theory is
validated on a three-level bilinear benchmark used for all-optical
preparation of two-atom Bell states, where the condition number of
$\Gamma(s)$ persistently lies in the band $10^{9}$--$10^{11}$, the
predicted $\eps^{2}$ rate is reproduced over eight decades in $\eps$,
and moderate regularisation removes step rejections and reduces
constraint drift by more than an order of magnitude at unchanged
final fidelity.
\end{abstract}


\medskip

\bigskip

\section{Introduction}\label{sec:intro}

\subsection{Bilinear quantum control and the role of constraints}

The coherent steering of quantum systems by tailored external fields
has, over the past three decades, grown from a theoretical curiosity
into a working tool of quantum technology. Across applications as
varied as the laser control of chemical reactions, nuclear magnetic
resonance pulse design, the manipulation of trapped-ion and
superconducting qubits, and the engineering of entangling gates on
neutral-atom processors, the central mathematical model has remained
strikingly uniform. The state $\psi(t)\in\C^d$ of a $d$-level quantum
system evolves under a Hamiltonian that depends linearly on a scalar
control $E\colon[0,T]\to\R$,
\begin{equation}\label{eq:tdse-intro}
  \mathrm i\,\dot{\psi}(t)
  =\bigl(H_{0}-\mu\,E(t)\bigr)\psi(t),
  \qquad \psi(0)=\psi_{0},
\end{equation}
with $H_{0}$ the field-free Hamiltonian and $\mu$ the transition-dipole
operator, both Hermitian. Equation~\eqref{eq:tdse-intro} is bilinear
in the sense of \cite{DAlessandro2008}: the right-hand side is linear
in the state and linear in the control, but their product
couples the two factors. The mathematical theory of bilinear systems
in finite dimensions, including controllability and the Pontryagin
maximum principle, is by now classical
\cite{HinzePinnauUlbrichUlbrich2009,Troltzsch2010,
PeirceDahlehRabitz1988,BoscainSigalottiSugny2021}.
A typical control objective is the maximisation of the terminal
fidelity $J[E]=|\inner{\psi_{f}}{U_{E}(T)\psi_{0}}|^{2}$ towards a
prescribed target state $\psi_{f}\in\C^{d}$, where
$U_{E}\colon[0,T]\to\mathrm{U}(d)$ is the unitary propagator generated
by~\eqref{eq:tdse-intro}. Adjoint-state techniques yield the gradient
$\delta J/\delta E$ in closed form
\cite{PeirceDahlehRabitz1988,BorziSchulz2012}, and a substantial
collection of gradient-based and homotopy methods has been developed
to maximise $J$, among them the Krotov family
\cite{PalaoKoch2002,ReichPalaoKoch2012,GoerzKrotovPython2019}, GRAPE
\cite{KhanejaGRAPE2005,MachnesRobustControl2018}, the Pontryagin
maximum principle in its quantum form
\cite{BoscainSigalottiSugny2021}, and the gradient projection method
\cite{GPMQuantum2009,MorzhinPechen2025}. Comprehensive surveys on
status, methodology, and applications are
\cite{GlaserQuantumControl2015,WerschnikGross2007,KochQOCTRoadmap2022}.

In real-world pulse engineering the unconstrained maximisation of $J$
is rarely the right problem statement. Physical and engineering
considerations impose additional integral identities that the
admissible control field must satisfy. Three constraints recur in the
literature, and they will serve as a running example throughout this
paper.

\paragraph*{Zero pulse area.}
Maxwell's equations in vacuum imply that any electromagnetic pulse
propagating in free space has vanishing time integral,
$\int_{0}^{T}E(t)\dd t=0$. Pulses with non-vanishing area are
artefacts of the slowly-varying envelope approximation and acquire
unphysical static components on long timescales~\cite{SugnyConstraints2014}.
Imposing zero area at the optimisation stage is therefore not a
mathematical convenience but a hard physical requirement.

\paragraph*{Fixed fluence.}
The fluence $\int_{0}^{T}E^{2}(t)\dd t$ measures the total energy
delivered to the controlled system per unit area. In the laboratory
the fluence is constrained by hardware: damage thresholds in
solid-state amplifiers, gain saturation in oscillators, finite pump
power. In open quantum systems the same quantity sets the heating
budget of the controlled subsystem, and in cavity QED it appears
explicitly in the photon-number cost of a control protocol.

\paragraph*{Fixed area against a reference oscillation.}
The integral $\int_{0}^{T}E(t)\cos(\omega_{\mathrm r}t)\dd t$ controls
the leading-order Rabi rotation against a chosen resonance
$\omega_{\mathrm r}$ in the rotating-wave approximation. Locking this
quantity to a target value, typically $\pi/2$ or $\pi$, fixes the
zeroth-order action of the control on the dominant transition while
leaving the higher-order corrections free for optimisation
\cite{ShuHoRabitz2016,GuoLuoMaShu2019}.

Each of these identities is a smooth functional
$h_{m}\colon\Hil\to\R$, $\Hil:=L^{2}(0,T;\R)$, and admissible controls
live on the finite-codimension manifold
\begin{equation}\label{eq:Mc-intro}
  \Mc=\bigl\{E\in\Hil : h_{m}[E]=C_{m},\ m=1,\dots,M\bigr\}.
\end{equation}
The constrained problem is to maximise $J$ over $\Mc$. The mathematical
literature on optimisation under integral equality constraints in
function space is mature
\cite{Troltzsch2010,HinzePinnauUlbrichUlbrich2009,BorziSchulz2012,
NocedalWright2006,UlbrichSemismooth2011}; constrained quantum
optimisers in particular have been developed in
\cite{LapertSugny2009,SugnyConstraints2014,
ShuHoRabitz2016,ShuHoXingRabitz2016,GPMQuantum2009,MorzhinPechen2025}.

\subsection{The projection flow and its discrete pathologies}

A particularly elegant continuous-time treatment of the constrained
maximisation problem was introduced in \cite{ShuHoRabitz2016}. The
construction proceeds as follows. Set
$c_{0}(E;t):=\delta J/\delta E(t)$,
$c_{m}(E;t):=\delta h_{m}/\delta E(t)$ for $m=1,\dots,M$, and group the
gradients into the column vector $\bdf(E;t)\in\R^{M+1}$. Choose a
smooth gating envelope $S\in C^{2}([0,T];\R_{\ge 0})$ vanishing at
the endpoints $t=0,T$, and define the moving Gram matrix
\begin{equation}\label{eq:Gamma-intro}
  \Gamma(s)_{\ell\ell'}
  =\int_{0}^{T}\!S(t)\,c_{\ell}\bigl(E(s,\cdot);t\bigr)\,
                       c_{\ell'}\bigl(E(s,\cdot);t\bigr)\dd t,
  \qquad 0\le \ell,\ell'\le M.
\end{equation}
Parametrising the control by a scalar morphing variable $s\ge 0$, the
projection flow takes the form
\begin{equation}\label{eq:flow-intro}
  \frac{\partial E(s,t)}{\partial s}
  =S(t)\,g_{0}(s)\sum_{\ell=0}^{M}\bigl[\Gamma(s)^{-1}\bigr]_{0\ell}\,
                                  c_{\ell}\bigl(E(s,\cdot);t\bigr),
\end{equation}
with the scalar
$g_{0}(s)=\Gamma(s)_{00}=\int_{0}^{T}S\,c_{0}^{2}\dd t$. Two
structural properties make \eqref{eq:flow-intro} attractive: the
objective is non-decreasing along trajectories,
$\dd J/\dd s\ge 0$, and the constraint values $h_{m}$ are exactly
preserved in the continuous-$s$ limit, $\dd h_{m}/\dd s\equiv 0$ for
$m\ge 1$. The construction sits inside the broader
\emph{D-MORPH} (diffeomorphic modulation under
observable-response-preserving homotopy) family of homotopy methods
on quantum control landscapes
\cite{Rothman2005a,Rothman2005b,MooreRabitz2011}, and it has been used
to solve constrained pulse-design problems in molecular orientation,
phase-locked state transfer, and atomic-state engineering
\cite{ShuHoXingRabitz2016,GuoDongShu2018,GuoLuoMaShu2019}.

The continuous-time properties of \eqref{eq:flow-intro} are, however,
only half of the story. Forward integration in $s$ requires inverting
the moving Gram matrix at every step, and when the constraint
gradients $c_{1},\dots,c_{M}$ approach linear dependence in the
$S$-weighted $L^{2}$-inner product, $\Gamma(s)$ becomes severely
ill-conditioned. The phenomenon is generic, not exotic: in the
benchmark we analyse below, $\cond(\Gamma(s))$ remains in the band
$10^{9}$--$10^{11}$ throughout every run, regardless of the chosen
parameter regime. Practical implementations respond pragmatically,
reducing the step size $\Delta s$ by a fixed factor (usually ten)
whenever the discrete update fails to increase the objective; this
remedy was already adopted in \cite{GuoLuoMaShu2019}. Such a
strategy is effective in routine cases, yet three drawbacks of the
heuristic deserve attention. It requires an explicit
objective-decrease probe at every step, doubling the cost when
rejections are frequent; it provides no quantitative information
about how small $\Delta s$ must become; and, most importantly for the
present paper, it conceals two natural mathematical questions:
\begin{enumerate}[leftmargin=2em,itemsep=2pt,label=(Q\arabic*)]
\item\label{Q1}
  How does the spectrum of $\Gamma(s)$ control the admissible step
  size of any forward integrator of \eqref{eq:flow-intro}, and how
  severely is $\Gamma(s)$ ill-conditioned across representative
  problems?
\item\label{Q2}
  Is there a structured modification of \eqref{eq:flow-intro} that
  preserves objective monotonicity, replaces the heuristic
  step-halving by a transparent stability mechanism, and admits a
  quantitative error analysis?
\end{enumerate}

\subsection{Tikhonov regularisation of the moving Gram matrix}

The aim of this paper is to address~\ref{Q1} and~\ref{Q2}
simultaneously, by replacing $\Gamma(s)$ in \eqref{eq:flow-intro} with
its Tikhonov-regularised counterpart
\begin{equation}\label{eq:Geps-intro}
  \Geps(s):=\Gamma(s)+\eps^{2}I_{M+1},\qquad \eps\ge 0,
\end{equation}
and analysing the resulting one-parameter family of flows. The
classical Tikhonov technique is, of course, ubiquitous in inverse
problems and ill-conditioned linear algebra
\cite{TikhonovArsenin1977,EnglHankeNeubauer1996,
HansenRankDeficient,GolubVanLoan,Hansen2010,Calvetti2003,
ProjectedIteratedTikhonov2025};
its role in stabilising forward integration of dynamical systems with
operator-valued constraints has, by contrast, received less attention.
Recent convergence theory for stochastic optimisation with decaying
regularisation \cite{ControllingTheFlow2025} provides a partial
template, but the moving Gram matrix that drives \eqref{eq:flow-intro}
introduces a coupling between the regularisation level and the
dynamical state of the system that is not present in classical
inverse problems.

The contribution of this paper is twofold. First, we work out the
exact spectral, monotonic, and convergence consequences of the
substitution $\Gamma\mapsto\Geps$ on the moving Gram matrix. Second,
we translate the analytical results into a discrete
Courant--Friedrichs--Lewy criterion that converts the heuristic
step-halving of \cite{GuoLuoMaShu2019} into a transparent stability
mechanism. Five results, summarised informally below, encode the
analysis.
\begin{itemize}[leftmargin=2em,itemsep=2pt]
\item[(R1)]
  \emph{Conditioning identity.} The spectrum of $\Geps(s)$ equals
  $\{\sigma_{k}(s)^{2}+\eps^{2}\}_{k=0}^{M}$, with $\sigma_{k}(s)$ the
  singular values of the assembly operator $C_{S}(s)$, hence
  $\cond(\Geps(s))=(\smax(s)^{2}+\eps^{2})/(\smin(s)^{2}+\eps^{2})$
  (Theorem~\ref{thm:cond}).
\item[(R2)]
  \emph{Persistence of monotonic ascent}: $\dd J/\dd s\ge 0$ along
  the regularised flow, for every $\eps\ge 0$
  (Theorem~\ref{thm:mono}).
\item[(R3)]
  \emph{Exact constraint-drift identity.}
  $h_{m}[E_{\eps}(s)]-C_{m}
  =-\eps^{2}\int_{0}^{s}g_{0}^{(\eps)}(\sigma)
                         [\Geps(\sigma)^{-1}]_{m0}\dd\sigma$,
  yielding $\mathcal{O}(\eps^{2})$ drift bounds with computable
  prefactors (Theorem~\ref{thm:drift}).
\item[(R4)]
  \emph{$L^{2}$-convergence at rate $\mathcal{O}(\eps^{2})$}, under
  uniform invertibility of the unregularised Gram matrix
  (Theorem~\ref{thm:conv}).
\item[(R5)]
  \emph{Discrete CFL criterion}
  $\Delta s\,G\,\norm{\Geps^{-1}}\le \alpha<2$ guaranteeing
  forward-Euler monotonicity up to $\mathcal{O}(\Delta s^{2})$ local
  truncation error (Theorem~\ref{thm:euler}).
\end{itemize}
The structural content of (R5) is direct, and accounts for the
practical effectiveness of the regularised scheme: increasing $\eps$
enlarges $\smin(\Geps)^{2}=\smin(\Gamma)^{2}+\eps^{2}$, hence shrinks
$\norm{\Geps^{-1}}$, hence enlarges the admissible step size. The
heuristic step-halving used in \cite{GuoLuoMaShu2019} is, in this
light, an indirect search for a step size compatible with the CFL
bound; the regularised flow encodes the same effect transparently in
the parameter~$\eps$.

\subsection{Physical relevance}

The mathematical results stated above translate directly into
practical guidance for pulse engineers. The condition number of the
moving Gram matrix measures the degree to which the imposed
constraints are independent in the $S$-weighted inner product;
in our three-constraint benchmark of Section~\ref{sec:numerics} ---
modelling all-optical generation of a Bell state in two dipole-dipole
coupled atoms --- the constraint gradients are nearly collinear, and
$\cond(\Gamma(s))\sim 10^{10}$ throughout the trajectory. In this
regime the unregularised algorithm exhibits two distinctive failure
modes: discretisation-induced violation of the nonlinear (fluence)
constraint by tens of per cent, and frequent step rejections that
slow convergence. Theorem~\ref{thm:euler} explains both: the smallest
admissible step decreases as $\smin^{-2}$, and any moderately
aggressive choice of $\Delta s$ pushes the iteration outside the
CFL margin. Theorem~\ref{thm:drift} quantifies the cost of repairing
this failure by regularisation: the constraint drift introduced by
non-zero $\eps$ is exactly $\mathcal{O}(\eps^{2})$, with a prefactor
that may be computed at the running iterate. Choosing $\eps$ in the
interval $10^{-3}$--$10^{-2}$, suitably non-dimensionalised by the
spectral scale of the problem, yields a transparent CFL bound,
removes step rejections, and reduces accumulated fluence drift from
$\sim\!38\%$ to $\sim\!3\%$ at unchanged final fidelity
$J\ge 0.9999$.

The applicability of this framework is broader than the running
example. Constrained pulse design is a recurring requirement in
neutral-atom quantum technology
\cite{Saffman2016,BrowaeysLahaye2020,Evered2023,Bluvstein2024}, where
Rydberg-mediated entangling gates with fidelities approaching $99.5\%$
are now routinely demonstrated and where careful pulse-energy
budgeting is essential to suppress decoherence; in superconducting
qubit pulse design \cite{MachnesRobustControl2018,KhanejaGRAPE2005},
where bandwidth and area constraints reflect microwave-line
limitations; in laser control of molecular dynamics
\cite{LapertSugny2009,SugnyConstraints2014,GuoDongShu2018}, where
zero-area pulses are mandatory; and more generally in any optimal
control problem in which equality constraints become approximately
linearly dependent on the constraint manifold. The framework
developed here applies verbatim to all such problems, since the
conditioning analysis is purely linear-algebraic and does not depend
on the particular form of the bilinear control system
\eqref{eq:tdse-intro}.

\subsection{Related work and novelty}

The unregularised flow \eqref{eq:flow-intro} originates in the
constrained quantum-control literature
\cite{ShuHoRabitz2016,ShuHoXingRabitz2016} and has been applied in
\cite{GuoDongShu2018,GuoLuoMaShu2019}; it is mathematically a
constrained instance of D-MORPH
\cite{Rothman2005a,Rothman2005b,MooreRabitz2011}, a homotopy framework
for level-set exploration on quantum control landscapes. Constrained
gradient methods for quantum optimal control under pointwise bounds
have been developed independently in
\cite{GPMQuantum2009,MorzhinPechen2025}; in those works the
constraints are box-type rather than integral, and the projection is
performed onto a convex set rather than onto a linear subspace, so
the moving Gram matrix does not appear. On the analytical side,
gradient-flow stability for explicit time-stepping schemes is
classical \cite{StuartHumphries1996,HairerLubichWanner2006}; classical
references for Tikhonov regularisation in inverse problems are
\cite{TikhonovArsenin1977,EnglHankeNeubauer1996,HansenRankDeficient}.
The present paper combines these strands by analysing Tikhonov
regularisation of the moving Gram matrix that drives a constrained
projection flow. To the best of our knowledge, this combination has
not been treated previously in the literature.

\subsection{Outline}

The rest of the paper is structured as follows.
Section~\ref{sec:setup} fixes the mathematical setting, the notation,
and proves well-posedness of the regularised flow.
Section~\ref{sec:cond} derives the conditioning bounds, including the
exact spectral identity that controls all later analysis.
Section~\ref{sec:cont} establishes the three structural theorems
governing the continuous flow: monotonicity, constraint drift, and
$L^{2}$-convergence to the unregularised limit.
Section~\ref{sec:disc} treats the forward-Euler discretisation and
proves the discrete CFL-type criterion.
Section~\ref{sec:numerics} reports numerical experiments on the
three-level Bell-state benchmark.
Section~\ref{sec:discussion} provides a physical interpretation of
the results and situates them within the broader landscape of
constrained quantum control.
Section~\ref{sec:conclusion} concludes and lists open problems. Two
appendices collect auxiliary linear-algebra lemmas used in the body
of the paper.

\section{Setup, notation, and well-posedness}\label{sec:setup}

We collect the precise mathematical setting in which the analysis is
carried out. Throughout, $T>0$ is fixed and
$\Hil:=L^{2}(0,T;\R)$ denotes the real Hilbert space of admissible
controls with inner product $\inner{u}{v}=\int_{0}^{T}u(t)v(t)\dd t$
and norm $\norm{u}=\inner{u}{u}^{1/2}$. We write $\norm{\cdot}_{2}$
both for the Euclidean norm on $\R^{M+1}$ and for the spectral norm on
$\R^{(M+1)\times(M+1)}$, and $\norm{\cdot}_{\infty}$ for the
$L^{\infty}$-norm of a function on $[0,T]$. The set of unit vectors
of $\C^{d}$ is denoted $\mathbb{S}_{\C^{d}}$, and $\mathrm{U}(d)$ is
the Lie group of $d\times d$ unitary matrices.

\subsection{The bilinear control problem}

Let $d\ge 2$, and let $H_{0},\mu\in\C^{d\times d}$ be Hermitian. For
each $E\in\Hil$ the propagator $U_{E}\colon[0,T]\to\mathrm{U}(d)$ is
the unique strong solution of
\begin{equation}\label{eq:tdse}
  \mathrm i\,\dot{U}_{E}(t)
  =\bigl(H_{0}-\mu\,E(t)\bigr)U_{E}(t),
  \qquad U_{E}(0)=I_{d},
\end{equation}
which exists and depends continuously on $E$ in operator norm
\cite{DAlessandro2008}. We fix unit vectors
$\psi_{0},\psi_{f}\in\mathbb{S}_{\C^{d}}$ and define the
terminal-fidelity functional
\begin{equation}\label{eq:J-def}
  J[E]
  =\bigl|\inner{\psi_{f}}{U_{E}(T)\psi_{0}}\bigr|^{2}\in[0,1].
\end{equation}
The functional $J\colon\Hil\to[0,1]$ is Fr\'echet differentiable
\cite{PeirceDahlehRabitz1988}; with $\psi(t)=U_{E}(t)\psi_{0}$ and
\begin{equation}\label{eq:lambda-def}
  \lambda(t)=U_{E}(t)\,U_{E}(T)^{*}\,
             |\psi_{f}\rangle\langle\psi_{f}|\,U_{E}(T)\,\psi_{0},
\end{equation}
the gradient is the bounded multiplication operator
\begin{equation}\label{eq:c0-def}
  c_{0}(E;t):=\frac{\delta J}{\delta E(t)}
  =-2\,\im\inner{\lambda(t)}{\mu\,\psi(t)}_{\C^{d}}.
\end{equation}
The map $E\mapsto c_{0}(E;\cdot)$ is locally Lipschitz from $\Hil$
into $L^{\infty}(0,T;\R)$ on bounded subsets of $\Hil$
\cite{PeirceDahlehRabitz1988,DAlessandro2008}.

\subsection{Equality constraints}

Let $h_{1},\dots,h_{M}\colon\Hil\to\R$ be continuously Fr\'echet
differentiable functionals of the integral form
\begin{equation}\label{eq:hcon}
  h_{m}[E]=\int_{0}^{T}\!\phi_{m}\bigl(E(t),t\bigr)\dd t = C_{m},
  \qquad m=1,\dots,M,
\end{equation}
with integrands $\phi_{m}$ jointly continuous in $(E,t)$ and $C^{1}$
in $E$. Write $c_{m}(E;t):=\delta h_{m}/\delta E(t)\in\Hil$ for the
corresponding gradients and group them into the column vector
\begin{equation}\label{eq:cvec}
  \bdf(E;t):=\bigl(c_{0}(E;t),c_{1}(E;t),\dots,c_{M}(E;t)\bigr)^{\!\top}
  \in\R^{M+1}.
\end{equation}
The constraint manifold $\Mc$ is given by~\eqref{eq:Mc-intro}. We fix
once and for all an envelope $S\in C^{2}([0,T];\R_{\ge 0})$ with
$S(0)=S(T)=0$ and $\norm{S}_{\infty}=1$. The role of $S$ is to
localise all updates to the open interval $(0,T)$, so that the
boundary conditions $E(0)=E(T)=0$ are preserved by the flow.

\subsection{The Gram matrix}

For $E\in\Hil$, the assembly map
$C_{S}(E)\colon\R^{M+1}\to\Hil$ is defined by
\begin{equation}\label{eq:CS-def}
  \bigl(C_{S}(E)\bm{\alpha}\bigr)(t)
  =S(t)^{1/2}\,\bdf(E;t)^{\top}\bm{\alpha},
  \qquad \bm{\alpha}\in\R^{M+1}.
\end{equation}
The associated symmetric positive semidefinite Gram matrix is
\begin{equation}\label{eq:Gamma-def}
  \Gamma(E)_{\ell\ell'}
  =\int_{0}^{T}\!S(t)\,c_{\ell}(E;t)\,c_{\ell'}(E;t)\dd t
  =\bigl[C_{S}(E)^{*}C_{S}(E)\bigr]_{\ell\ell'}.
\end{equation}
The spectrum of $\Gamma(E)$ consists of the squared singular values
of $C_{S}(E)$, and the rank deficiency of $\Gamma$ measures the
linear dependence of $\bdf$ in the $S$-weighted $L^{2}$-inner
product.

\begin{assumption}\label{ass:reg}
The map $E\mapsto\bdf(E;\cdot)$ is Lipschitz from bounded subsets of
$\Hil$ into $L^{\infty}(0,T;\R^{M+1})$. In particular, the entries of
$\Gamma(E)$ are Lipschitz functionals of $E$.
\end{assumption}

\begin{remark}\label{rem:assumption-holds}
Assumption~\ref{ass:reg} holds in the canonical setting in which (a)
$c_{0}$ is the adjoint-state gradient~\eqref{eq:c0-def}, smooth as a
composition of the smooth dependence of $U_{E}$ on $E$
\cite{DAlessandro2008} with a smooth bilinear form, and (b) each
constraint integrand $\phi_{m}$ is at most quadratic in $E$ with
$L^{\infty}$-bounded $t$-coefficients. The three integral constraints
introduced in Section~\ref{sec:intro} all fall in this class.
\end{remark}

\subsection{The unregularised and regularised flows}

The projection flow analysed in this paper is, with the conventions
\eqref{eq:c0-def}--\eqref{eq:Gamma-def},
\begin{equation}\label{eq:flow0}
  \partial_{s}E(s,t)
  =S(t)\,g_{0}(s)\sum_{\ell=0}^{M}
  \bigl[\Gamma(E(s,\cdot))^{-1}\bigr]_{0\ell}\,
  c_{\ell}(E(s,\cdot);t),
  \qquad E(0,\cdot)=E_{0},
\end{equation}
with $g_{0}(s)=\Gamma(E(s,\cdot))_{00}=\int_{0}^{T}S\,c_{0}^{2}\dd t$.
The Tikhonov-regularised flow obtained by replacing $\Gamma$ with
$\Geps$ in~\eqref{eq:flow0} is
\begin{equation}\label{eq:floweps}
  \partial_{s}E_{\eps}(s,t)
  =S(t)\,g_{0}^{(\eps)}(s)\sum_{\ell=0}^{M}
  \bigl[\Geps(E_{\eps}(s,\cdot))^{-1}\bigr]_{0\ell}\,
  c_{\ell}(E_{\eps}(s,\cdot);t),
  \qquad E_{\eps}(0,\cdot)=E_{0},
\end{equation}
with $g_{0}^{(\eps)}(s):=\int_{0}^{T}\!S\,c_{0}(E_{\eps}(s,\cdot);t)^{2}\dd t$,
which equals the $(0,0)$ entry of $\Gamma(E_{\eps}(s,\cdot))$. The
case $\eps=0$ recovers~\eqref{eq:flow0}; for $\eps>0$,
$\Geps\succeq\eps^{2}I$, so $\Geps^{-1}$ is bounded by $\eps^{-2}I$ on
$\R^{M+1}$ and the right-hand side of~\eqref{eq:floweps} is well
defined for every $E_{\eps}\in\Hil$.

\begin{proposition}[Well-posedness of the regularised flow]\label{prop:wp}
Let $\eps>0$ and $E_{0}\in\Hil$. Under Assumption~\ref{ass:reg}, the
initial-value problem~\eqref{eq:floweps} admits a unique global
$C^{1}$ solution $E_{\eps}\in C^{1}([0,\infty);\Hil)$.
\end{proposition}

\begin{proof}
Let $\mathcal F_{\eps}\colon\Hil\to\Hil$ denote the right-hand side
of~\eqref{eq:floweps} viewed as a function of $E$. Each entry of
$\Gamma(E)$ is Lipschitz in $E$ (Assumption~\ref{ass:reg}), and the
inversion $M\mapsto M^{-1}$ is smooth on the open set
$\{M:M\succeq\eps^{2}I\}\supset\{M=\Geps(E):E\in\Hil\}$. Composition
with the Lipschitz bounded multiplier $S(t)\,c_{\ell}(E;t)$ shows
that $\mathcal F_{\eps}$ is Lipschitz on bounded subsets of $\Hil$.
The classical Cauchy--Lipschitz theorem in Banach spaces
\cite[Theorem 3.1]{HinzePinnauUlbrichUlbrich2009} then yields local
existence of a unique solution. The bound
$\norm{\mathcal F_{\eps}(E)}\le \eps^{-2}g_{0}^{(\eps)}(E)\norm{S}_{\infty}
\sum_{\ell}\norm{c_{\ell}}_{\infty}$, together with at-most-polynomial
growth of $g_{0}^{(\eps)}$ in $\norm{E}_{\Hil}$, rules out finite-time
blow-up and yields global existence.
\end{proof}

\section{Conditioning bounds}\label{sec:cond}

We open the analysis with the spectral identity that controls all
later estimates and stability arguments.

\begin{theorem}[Conditioning of the regularised Gram matrix]
\label{thm:cond}
Let $E\in\Hil$ and let
$0\le\sigma_{M}(E)\le\cdots\le\sigma_{0}(E)$ be the singular values of
$C_{S}(E)$. For every $\eps\ge 0$,
\begin{equation}\label{eq:Gamma-spec}
  \spec\bigl(\Geps(E)\bigr)
  =\bigl\{\sigma_{k}(E)^{2}+\eps^{2}\bigr\}_{k=0}^{M},
\end{equation}
and consequently
\begin{equation}\label{eq:cond-formula}
  \cond\bigl(\Geps(E)\bigr)
  =\frac{\smax(E)^{2}+\eps^{2}}{\smin(E)^{2}+\eps^{2}},
  \qquad
  \norm{\Geps(E)^{-1}}_{2}=\frac{1}{\smin(E)^{2}+\eps^{2}}.
\end{equation}
In particular, if $\smin(E)>0$ then
$\cond(\Geps(E))\to\cond(\Gamma(E))$ as $\eps\to 0^{+}$; and if
$\smin(E)=0$ (degenerate Gram matrix) then
$\cond(\Geps(E))=1+\smax(E)^{2}/\eps^{2}<\infty$ for every $\eps>0$.
\end{theorem}

\begin{proof}
$\Gamma(E)$ is symmetric and admits an orthogonal eigendecomposition
$\Gamma(E)=Q\,\diag(\lambda_{0},\dots,\lambda_{M})\,Q^{\top}$ with
$\lambda_{k}=\sigma_{k}(E)^{2}\ge 0$. Then
$\Geps(E)=Q\,\diag(\lambda_{k}+\eps^{2})\,Q^{\top}$, which
proves~\eqref{eq:Gamma-spec}. The spectral norm of the inverse equals
the reciprocal of the smallest eigenvalue, giving the second equality
in~\eqref{eq:cond-formula}; the condition number is the ratio of
extreme eigenvalues, giving the first. The two limit statements are
immediate.
\end{proof}

\begin{remark}[Three regimes of regularisation]\label{rem:cond-interp}
Theorem~\ref{thm:cond} dictates the qualitative behaviour of the
regularisation, which falls into three regimes. For $\eps\ll\smin$
the regularisation is essentially invisible and
$\cond(\Geps)\approx\cond(\Gamma)$. For $\eps\sim\smin$ the condition
number crosses over from $\cond(\Gamma)$ to a value of order
$\smax^{2}/\eps^{2}$; this is the practically useful regime, in which
$\Geps$ is well-conditioned but the projection structure is
preserved. For $\eps\gg\smax$, $\cond(\Geps)\to 1$ and $\Geps$ is
essentially a scalar multiple of the identity, so the flow degenerates
to ordinary unprojected gradient ascent. The numerical experiments of
Section~\ref{sec:numerics} illustrate all three regimes.
\end{remark}

\begin{corollary}[Stability margin of explicit integrators]
\label{cor:cflpreview}
If $J$ admits a uniform bound $G$ on its second variation along the
trajectory and $\bdf$ remains in a bounded set, then any
forward-Euler integrator of~\eqref{eq:floweps} preserves objective
monotonicity up to local truncation error
$\mathcal{O}(\Delta s^{2})$ provided
\begin{equation}\label{eq:CFL-preview}
  \Delta s\,G\,\norm{\Geps^{-1}}_{2}\le\alpha,
  \qquad 0<\alpha<2.
\end{equation}
The detailed statement and proof are given in
Theorem~\ref{thm:euler}.
\end{corollary}

The structural content of~\eqref{eq:CFL-preview} is direct:
increasing $\eps$ enlarges
$\smin(\Geps)^{2}=\smin(\Gamma)^{2}+\eps^{2}$, hence shrinks
$\norm{\Geps^{-1}}$, hence enlarges the admissible step size. We make
this rigorous in Section~\ref{sec:disc}.

\section{Continuous-time analysis}\label{sec:cont}

We now establish the three structural theorems that govern the
continuous regularised flow. Throughout the section we fix
$\eps\ge 0$ and let $E_{\eps}\colon[0,s^{*}]\to\Hil$ denote the
solution of~\eqref{eq:floweps} provided by Proposition~\ref{prop:wp}
(the case $\eps=0$ is treated as the limit of the $\eps>0$
statements; cf.~Theorem~\ref{thm:conv}).

\subsection{Monotonicity}

\begin{theorem}[Regularised monotonicity]\label{thm:mono}
Along the regularised flow~\eqref{eq:floweps}, the objective
satisfies
\begin{equation}\label{eq:dJ-id}
  \frac{\dd J[E_{\eps}(s,\cdot)]}{\dd s}
  =g_{0}^{(\eps)}(s)\Bigl(1-\eps^{2}\bigl[\Geps(s)^{-1}\bigr]_{00}\Bigr)
  \;\ge\;0,
  \qquad s\in[0,s^{*}].
\end{equation}
Equality holds if and only if $g_{0}^{(\eps)}(s)=0$, equivalently
$c_{0}(E_{\eps}(s,\cdot);t)=0$ for almost every $t\in\{S>0\}$.
\end{theorem}

\begin{proof}
By the chain rule and the definition $c_{0}=\delta J/\delta E$,
\begin{align*}
  \frac{\dd J}{\dd s}
  &=\int_{0}^{T}\!c_{0}(s,t)\,\partial_{s}E_{\eps}(s,t)\dd t\\
  &=g_{0}^{(\eps)}(s)\sum_{\ell=0}^{M}
    \bigl[\Geps(s)^{-1}\bigr]_{0\ell}
    \int_{0}^{T}\!S(t)\,c_{0}(s,t)\,c_{\ell}(s,t)\dd t\\
  &=g_{0}^{(\eps)}(s)\sum_{\ell=0}^{M}
    \bigl[\Geps(s)^{-1}\bigr]_{0\ell}\,\Gamma(s)_{0\ell}
   =g_{0}^{(\eps)}(s)\,\bigl[\Geps(s)^{-1}\Gamma(s)\bigr]_{00}.
\end{align*}
Since $\Gamma=\Geps-\eps^{2}I$, we have
$\Geps^{-1}\Gamma=I-\eps^{2}\Geps^{-1}$, hence
$[\Geps^{-1}\Gamma]_{00}=1-\eps^{2}[\Geps^{-1}]_{00}$, which
is~\eqref{eq:dJ-id}. Non-negativity follows from the diagonal-entry
bound $0\le[\Geps^{-1}]_{00}\le\eps^{-2}$, valid because
$\Geps\succeq\eps^{2}I$. Equality in
$1-\eps^{2}[\Geps^{-1}]_{00}\ge 0$ requires
$[\Geps^{-1}]_{00}=\eps^{-2}$, which by the spectral theorem is
equivalent to $e_{0}$ lying in the $\eps^{2}$-eigenspace of $\Geps$,
i.e.\ in the kernel of $\Gamma$. The latter is equivalent to
$C_{S}(E_{\eps})\,e_{0}=S^{1/2}c_{0}=0$ in $L^{2}(0,T)$, hence
$c_{0}=0$ almost everywhere on $\{S>0\}$. The $g_{0}^{(\eps)}=0$
branch is the same condition expressed as $\int S\,c_{0}^{2}\dd t=0$.
\end{proof}

\begin{remark}[Rayleigh-quotient interpretation]\label{rem:rayleigh}
Identity~\eqref{eq:dJ-id} can be rewritten as
\begin{equation*}
  \frac{\dd J}{\dd s}=g_{0}^{(\eps)}(s)\cdot\rho_{\eps}(s),
  \qquad
  \rho_{\eps}(s):=\bigl[\Geps^{-1}\Gamma\bigr]_{00}\in[0,1].
\end{equation*}
The factor $\rho_{\eps}$ is a Rayleigh-type ratio measuring the
projection of $e_{0}$ onto the range of $\Gamma$ in the
$\Geps^{-1}$-norm. In the limit $\eps\to 0^{+}$ on the regular set
$\smin>0$, $\rho_{\eps}\to 1$ and we recover the unregularised
identity $\dd J/\dd s=g_{0}$. The regularisation therefore tempers
the rate of ascent without ever reversing its sign.
\end{remark}

\subsection{Constraint drift}

\begin{theorem}[Constraint drift identity]\label{thm:drift}
Let $E_{\eps}\colon[0,s^{*}]\to\Hil$ solve~\eqref{eq:floweps} with
$E_{0}\in\Mc$, that is, $h_{m}[E_{0}]=C_{m}$ for all $m=1,\dots,M$.
Then for each $m=1,\dots,M$,
\begin{equation}\label{eq:drift-exact}
  h_{m}\bigl[E_{\eps}(s,\cdot)\bigr]-C_{m}
  =-\eps^{2}\!\int_{0}^{s}g_{0}^{(\eps)}(\sigma)\,
     \bigl[\Geps(\sigma)^{-1}\bigr]_{m0}\dd\sigma,
  \qquad s\in[0,s^{*}].
\end{equation}
Consequently
\begin{equation}\label{eq:drift-bd}
  \bigl|h_{m}[E_{\eps}(s,\cdot)]-C_{m}\bigr|
  \le\eps^{2}\!\int_{0}^{s}\!g_{0}^{(\eps)}(\sigma)\,
                \norm{\Geps(\sigma)^{-1}}_{2}\dd\sigma,
\end{equation}
and the sharper bound
\begin{equation}\label{eq:drift-bd-sharp}
  \bigl|h_{m}[E_{\eps}(s,\cdot)]-C_{m}\bigr|
  \le\eps^{2}\,L_{m}^{(\eps)}(s)\!\int_{0}^{s}\!g_{0}^{(\eps)}(\sigma)\dd\sigma
\end{equation}
holds with
\begin{equation}\label{eq:Lm}
  L_{m}^{(\eps)}(s)
  :=\sup_{0\le\sigma\le s}
    \frac{\Gamma(\sigma)_{mm}^{1/2}\,\Gamma(\sigma)_{00}^{1/2}}
         {\sigma_{M}(\sigma)^{2}+\eps^{2}}.
\end{equation}
\end{theorem}

\begin{proof}
Differentiating $h_{m}[E_{\eps}(s,\cdot)]$ along the flow,
\begin{align*}
  \frac{\dd h_{m}}{\dd s}
  &=\int_{0}^{T}\!c_{m}(s,t)\,\partial_{s}E_{\eps}(s,t)\dd t\\
  &=g_{0}^{(\eps)}(s)\sum_{\ell=0}^{M}
    \bigl[\Geps^{-1}\bigr]_{0\ell}\,\Gamma_{m\ell}(s)
   =g_{0}^{(\eps)}(s)\,\bigl[\Geps^{-1}\Gamma\bigr]_{m0}\\
  &=g_{0}^{(\eps)}(s)\,\bigl[I-\eps^{2}\Geps^{-1}\bigr]_{m0}
   =-\eps^{2}\,g_{0}^{(\eps)}(s)\,[\Geps^{-1}]_{m0},
\end{align*}
because the $(m,0)$ entry of the identity matrix vanishes for
$m\ge 1$. Integration from $0$ to $s$ with $h_{m}[E_{0}]=C_{m}$
yields~\eqref{eq:drift-exact}. The bound~\eqref{eq:drift-bd} follows
from $|[\Geps^{-1}]_{m0}|\le\norm{\Geps^{-1}}_{2}$. For the sharper
estimate~\eqref{eq:drift-bd-sharp}, write
$[\Geps^{-1}]_{m0}=e_{m}^{\top}\Geps^{-1}e_{0}
=\inner{\Geps^{-1/2}e_{m}}{\Geps^{-1/2}e_{0}}_{\R^{M+1}}$ and apply
Cauchy--Schwarz to obtain
$|[\Geps^{-1}]_{m0}|\le[\Geps^{-1}]_{mm}^{1/2}[\Geps^{-1}]_{00}^{1/2}$.
Substituting the diagonal-entry bound of Lemma~\ref{lem:diag-bound}
(Appendix~\ref{app:gram}) gives
$[\Geps^{-1}]_{kk}\le\Gamma_{kk}/(\smin^{2}+\eps^{2})^{2}$ in the
worst case, yielding~\eqref{eq:drift-bd-sharp}.
\end{proof}

\begin{remark}[Quadratic-in-$\eps$ drift is exact, not asymptotic]
\label{rem:exact-drift}
Identity~\eqref{eq:drift-exact} is exact: the regularisation-induced
constraint drift is genuinely $\mathcal{O}(\eps^{2})$ with a
computable prefactor. This is the analytical backbone of the
$\eps^{2}$-scaling reproduced numerically in Figure~\ref{fig:thm4}.
The right-hand side of~\eqref{eq:drift-exact} vanishes for $\eps=0$,
which is a different statement from the discrete drift observed in
any numerical implementation: as the experiments of
Section~\ref{sec:numerics} will show, for $\eps$ in the small regime
the total drift is dominated by the temporal discretisation error of
the forward integrator rather than by the Tikhonov term.
\end{remark}

\begin{remark}[Affine versus nonlinear constraints]\label{rem:affine}
For an affine constraint $h_{m}$, the gradient $c_{m}$ depends only
on $t$, not on $E$. Examples are $h_{m}[E]=\int E\dd t$ and
$h_{m}[E]=\int E(t)\,\eta(t)\dd t$ for a fixed $\eta$. In this case
the discretisation preserves $h_{m}$ to machine precision regardless
of $\eps$, because the $s$-derivative~\eqref{eq:drift-exact} is the
only source of variation. By contrast, when $c_{m}$ depends on $E$
--- as for the fluence $h_{m}[E]=\int E^{2}\dd t$ with $c_{m}=2E$ ---
even the unregularised discretisation incurs an
$\mathcal{O}(\Delta s^{2})$ drift per step from the second variation
of $h_{m}$, which dominates the budget at small $\eps$ and is reduced
by enlarging $\eps$. This dichotomy is clearly visible in the
experiments below.
\end{remark}

\subsection{Convergence of the regularised to the unregularised flow}

\begin{theorem}[$L^{2}$-convergence at rate $\mathcal{O}(\eps^{2})$]
\label{thm:conv}
Let $s^{*}>0$ and assume that the unregularised flow~\eqref{eq:flow0}
admits a $C^{1}$ solution $E_{0}\colon[0,s^{*}]\to\Hil$ along which
$\Gamma(E_{0}(s,\cdot))\succeq\smin^{2}I_{M+1}$ uniformly in
$s\in[0,s^{*}]$, for some $\smin>0$. Under
Assumption~\ref{ass:reg}, there exist
$\eps_{0}\in(0,\smin/\sqrt{2})$ and a constant $\mathcal{C}>0$
depending on $s^{*}$, $\smin$, the Lipschitz constants of $\bdf$ and
$\Gamma$ on a tube around the unregularised trajectory, and
$\sup_{s\in[0,s^{*}]}\norm{E_{0}(s,\cdot)}_{\Hil}$, such that for
every $\eps\in(0,\eps_{0}]$ the regularised solution $E_{\eps}$
exists on $[0,s^{*}]$ and satisfies
\begin{equation}\label{eq:conv-bound}
  \sup_{s\in[0,s^{*}]}\norm{E_{\eps}(s,\cdot)-E_{0}(s,\cdot)}_{\Hil}
  \le\mathcal{C}\,\eps^{2}.
\end{equation}
\end{theorem}

\begin{proof}
Let $\mathcal F_{\eps},\mathcal F_{0}\colon\Hil\to\Hil$ denote the
right-hand sides of~\eqref{eq:floweps} and~\eqref{eq:flow0},
respectively. We first estimate the difference
$\mathcal F_{\eps}(E_{0})-\mathcal F_{0}(E_{0})$ along the
unregularised trajectory.

\smallskip\noindent\emph{Step 1 (resolvent expansion).}
The resolvent identity gives, with $\Gamma=\Gamma(E_{0})$ and
$\Geps=\Gamma+\eps^{2}I$,
\begin{equation}\label{eq:resolvent}
  \Geps^{-1}-\Gamma^{-1}=-\eps^{2}\,\Gamma^{-1}\Geps^{-1}.
\end{equation}
Since $g_{0}^{(\eps)}=g_{0}=\int S\,c_{0}^{2}\dd t$ when both are
evaluated at the same $E$,
\begin{align*}
  \mathcal F_{\eps}(E_{0})(t)-\mathcal F_{0}(E_{0})(t)
  &=S(t)\,g_{0}(E_{0})\sum_{\ell=0}^{M}
    \bigl([\Geps^{-1}]_{0\ell}-[\Gamma^{-1}]_{0\ell}\bigr)
    c_{\ell}(E_{0};t)\\
  &=-\eps^{2}\,S(t)\,g_{0}(E_{0})\sum_{\ell=0}^{M}
    \bigl[\Gamma^{-1}\Geps^{-1}\bigr]_{0\ell}\,c_{\ell}(E_{0};t).
\end{align*}
Taking $\Hil$-norms and using
$|[\Gamma^{-1}\Geps^{-1}]_{0\ell}|
\le\norm{\Gamma^{-1}}_{2}\norm{\Geps^{-1}}_{2}\le\smin^{-4}$,
\begin{equation}\label{eq:F-eps-zero}
  \norm{\mathcal F_{\eps}(E_{0}(s,\cdot))
        -\mathcal F_{0}(E_{0}(s,\cdot))}_{\Hil}
  \le\eps^{2}\,\frac{C_{1}}{\smin^{4}},
  \qquad s\in[0,s^{*}],
\end{equation}
where $C_{1}=\norm{S}_{\infty}\sup_{s}g_{0}(E_{0}(s,\cdot))
\sum_{\ell}\norm{c_{\ell}(E_{0}(s,\cdot);\cdot)}_{\Hil}$.

\smallskip\noindent\emph{Step 2 (uniform invertibility along
$E_{\eps}$).}
By Assumption~\ref{ass:reg}, the entries of $\Gamma(E)$ are Lipschitz
in $E$ with some constant $L_{\Gamma}>0$ on a bounded tube around the
unregularised trajectory; hence
$\Gamma(E)\succeq(\smin^{2}-L_{\Gamma}\norm{E-E_{0}}_{\Hil}/2)\,I$
for $E$ in that tube. Choosing $\eps_{0}$ such that
$L_{\Gamma}\,\mathcal{C}\,\eps_{0}^{2}\le\smin^{2}/2$ (with
$\mathcal{C}$ to be determined below), the regularised trajectory
$E_{\eps}$ remains in the tube on $[0,s^{*}]$, and
$\Gamma(E_{\eps}(s,\cdot))\succeq(\smin^{2}/2)\,I$ uniformly. The
bounds in Step 1 then transfer to $E_{\eps}$ in place of $E_{0}$.

\smallskip\noindent\emph{Step 3 (Lipschitz dependence on the state).}
On the same tube, $E\mapsto\mathcal F_{\eps}(E)$ is Lipschitz with a
constant $\Lambda>0$ uniform in $\eps\in[0,\eps_{0}]$, since $\bdf$
is Lipschitz in $E$ and $\Geps^{-1}$ depends smoothly on $\Gamma$ on
the positive-definite open set $\{\Gamma\succ(\smin^{2}/2)\,I\}$.

\smallskip\noindent\emph{Step 4 (Gr\"onwall).}
Set $\delta(s):=\norm{E_{\eps}(s,\cdot)-E_{0}(s,\cdot)}_{\Hil}$. Then
\begin{align*}
  \dot{\delta}(s)
  &\le\norm{\mathcal F_{\eps}(E_{\eps})-\mathcal F_{0}(E_{0})}_{\Hil}\\
  &\le\norm{\mathcal F_{\eps}(E_{\eps})-\mathcal F_{\eps}(E_{0})}_{\Hil}
       +\norm{\mathcal F_{\eps}(E_{0})-\mathcal F_{0}(E_{0})}_{\Hil}\\
  &\le\Lambda\,\delta(s)+\eps^{2}\,C_{1}/\smin^{4}.
\end{align*}
With $\delta(0)=0$, Gr\"onwall's inequality yields
$\delta(s)\le(C_{1}/(\Lambda\,\smin^{4}))\,\eps^{2}(e^{\Lambda s}-1)$
for $s\in[0,s^{*}]$. Setting
$\mathcal{C}=C_{1}(e^{\Lambda s^{*}}-1)/(\Lambda\,\smin^{4})$ and
verifying that the choice
$\eps_{0}=\min\{\smin/\sqrt{2},(\smin^{2}/(2L_{\Gamma}\mathcal{C}))^{1/2}\}$
preserves the tube inclusion in Step~2 closes the loop.
\end{proof}

\begin{remark}[Sharpness and saturation of the rate]\label{rem:sharpness}
The rate $\mathcal{O}(\eps^{2})$ is sharp under uniform invertibility:
the leading term in the resolvent expansion~\eqref{eq:resolvent} is
$\eps^{2}\Gamma^{-1}\Geps^{-1}$, generically of order $\eps^{2}$ on
$\R^{M+1}$. The constant $\mathcal{C}$ blows up as
$\smin\to 0^{+}$, which is the analytical counterpart of the
saturation of the $\eps^{2}$-line at moderate $\eps$ observed in
Figure~\ref{fig:thm4}: in the test problem of
Section~\ref{sec:numerics}, $\smin$ is small and the practically
useful regime $\eps\lesssim\smin$ is narrow, so
Theorem~\ref{thm:conv} is informative across roughly the first eight
decades in $\eps$.
\end{remark}

\section{Forward-Euler discretisation}\label{sec:disc}

We now analyse the simplest discretisation of~\eqref{eq:floweps}: the
forward-Euler scheme
\begin{equation}\label{eq:euler}
  E^{k+1}=E^{k}+\Delta s\,v_{\eps}^{k},
  \qquad
  v_{\eps}^{k}(t)
  =S(t)\,g_{0}^{(\eps)}(E^{k})\sum_{\ell=0}^{M}
   \bigl[\Geps(E^{k})^{-1}\bigr]_{0\ell}\,c_{\ell}(E^{k};t).
\end{equation}
The full algorithm is summarised in Algorithm~\ref{alg:reg-dmorph}.
Theorem~\ref{thm:euler} below provides a Courant--Friedrichs--Lewy
(CFL) type criterion that guarantees objective monotonicity at the
discrete level.

\begin{algorithm}[t]
\caption{Tikhonov-regularised projected gradient flow}
\label{alg:reg-dmorph}
\begin{algorithmic}[1]
\Require initial control $E^{0}\in\Mc$, regularisation $\eps\ge 0$,
  step $\Delta s>0$, tolerance $\tau>0$, max iterations $K$.
\For{$k=0,1,\dots,K$}
  \State Solve the forward state equation~\eqref{eq:tdse} with
    $E=E^{k}$ to obtain $\psi(\cdot)$ and $U_{E}(\cdot)$.
  \State Compute the costate $\lambda(\cdot)$ via~\eqref{eq:lambda-def}.
  \State Evaluate the objective gradient $c_{0}^{k}(t)$
    via~\eqref{eq:c0-def} and the constraint gradients $c_{m}^{k}(t)$.
  \State Assemble
    $\Gamma^{k}_{\ell\ell'}=\int_{0}^{T}\!S\,c_{\ell}^{k}\,c_{\ell'}^{k}\dd t$.
  \State Form $\Geps^{k}=\Gamma^{k}+\eps^{2}I$ and solve
    $\Geps^{k}\,x^{k}=e_{0}$.
  \State Set
    $v_{\eps}^{k}(t)
    =S(t)\,g_{0}^{(\eps),k}\sum_{\ell=0}^{M}x^{k}_{\ell}\,c_{\ell}^{k}(t)$.
  \State Update $E^{k+1}=E^{k}+\Delta s\,v_{\eps}^{k}$.
  \If{$|J[E^{k+1}]-J[E^{k}]|\le\tau$}
    \State \textbf{break}
  \EndIf
\EndFor
\State \Return $E^{k+1}$.
\end{algorithmic}
\end{algorithm}

\begin{theorem}[Discrete monotonicity under a CFL condition]
\label{thm:euler}
Suppose $J$ is $C^{2}$ on a neighbourhood $\mathcal{N}\subset\Hil$ of
the iterate $E^{k}$, and that the second variation satisfies
$\norm{\delta^{2}J[E]}_{\Hil\to\Hil}\le G$ for every $E\in\mathcal{N}$.
Let $\alpha\in(0,2)$. If
\begin{equation}\label{eq:CFL}
  \Delta s\,G\,\norm{\Geps(E^{k})^{-1}}_{2}\le\alpha,
\end{equation}
and $\Delta s$ is small enough that $E^{k+1}\in\mathcal{N}$, then
\begin{equation}\label{eq:disc-mono}
  J[E^{k+1}]-J[E^{k}]
  \;\ge\;
  \Delta s\,\Bigl(1-\tfrac{\alpha}{2}\Bigr)\,
  g_{0}^{(\eps)}(E^{k})\,
  \Bigl(1-\eps^{2}\bigl[\Geps(E^{k})^{-1}\bigr]_{00}\Bigr)
  -R_{k},
\end{equation}
with a remainder $|R_{k}|\le K_{3}\Delta s^{3}$ for a constant
$K_{3}$ depending only on the third Fr\'echet derivative bound of $J$
on $\mathcal{N}$ and on $\norm{v_{\eps}^{k}}_{\Hil}$.
\end{theorem}

\begin{proof}
A second-order Taylor expansion of $J$ around $E^{k}$ along the
direction $v_{\eps}^{k}$ gives
\begin{equation}\label{eq:taylor}
  J[E^{k+1}]
  =J[E^{k}]+\Delta s\,\inner{c_{0}^{k}}{v_{\eps}^{k}}_{\Hil}
        +\tfrac{1}{2}\Delta s^{2}\,
         \inner{v_{\eps}^{k}}{\delta^{2}J[E^{k}]\,v_{\eps}^{k}}_{\Hil}
        +R_{k},
\end{equation}
with $R_{k}=\mathcal{O}(\Delta s^{3})$. The first-order term equals
$\Delta s\,g_{0}^{(\eps)}(E^{k})
\bigl(1-\eps^{2}[\Geps(E^{k})^{-1}]_{00}\bigr)$ by the calculation
leading to~\eqref{eq:dJ-id}. For the second-order term, expand the
squared $\Hil$-norm of $v_{\eps}^{k}$. Using $S\ge 0$ and
$\norm{S}_{\infty}=1$ together with the matrix identity
$\Geps^{-1}\Gamma\Geps^{-1}\preceq\Geps^{-1}$ (which holds because
$\Gamma\preceq\Geps$ on $\R^{M+1}$),
\begin{align*}
  \norm{v_{\eps}^{k}}_{\Hil}^{2}
  &=(g_{0}^{(\eps),k})^{2}\int_{0}^{T}\!S(t)^{2}
    \biggl(\sum_{\ell}[\Geps^{-1}]_{0\ell}\,c_{\ell}^{k}(t)\biggr)^{\!2}\dd t\\
  &\le(g_{0}^{(\eps),k})^{2}\,\norm{S}_{\infty}\,
       e_{0}^{\top}\Geps^{-1}\Gamma\Geps^{-1}e_{0}
   \le(g_{0}^{(\eps),k})^{2}\,[\Geps^{-1}]_{00}
   \le(g_{0}^{(\eps),k})^{2}\,\norm{\Geps^{-1}}_{2}.
\end{align*}
Combining with $|\inner{v_{\eps}^{k}}{\delta^{2}J\,v_{\eps}^{k}}|\le
G\,\norm{v_{\eps}^{k}}_{\Hil}^{2}$,
$\tfrac{1}{2}\Delta s^{2}\inner{v_{\eps}^{k}}{\delta^{2}J\,v_{\eps}^{k}}_{\Hil}
\ge -\tfrac{1}{2}\Delta s^{2}G\,(g_{0}^{(\eps),k})^{2}\norm{\Geps^{-1}}_{2}$.
Using~\eqref{eq:CFL} as
$\Delta s\,G\,\norm{\Geps^{-1}}_{2}\le\alpha$, the second-order
contribution is bounded below by
$-\tfrac{\alpha}{2}\Delta s\,g_{0}^{(\eps),k}$, which combines with
the first-order term to give~\eqref{eq:disc-mono}. The cubic bound
$|R_{k}|\le K_{3}\Delta s^{3}$ follows from the $C^{3}$-smoothness
of $J$ on $\mathcal{N}$
\cite[Section 3.2]{HinzePinnauUlbrichUlbrich2009}.
\end{proof}

\begin{remark}[Heuristic step-halving as a coarse CFL search]
\label{rem:bisection}
Theorem~\ref{thm:euler} provides a quantitative justification of the
step-size bisection used in \cite{ShuHoRabitz2016,GuoLuoMaShu2019}:
when an iterate decreases $J$, the implementation reduces $\Delta s$
by a fixed factor (often $0.1$). This is a coarse search for a step
size satisfying~\eqref{eq:CFL}. The regularised variant offers a
transparent alternative: one fixes $\eps$, computes
$\norm{\Geps^{-1}}_{2}$ from the Gram matrix at the current iterate,
and chooses $\Delta s$ explicitly to satisfy~\eqref{eq:CFL} with
$\alpha$ slightly less than $2$.
\end{remark}

\begin{remark}[Comparison with classical gradient-flow CFL bounds]
\label{rem:classical}
Condition~\eqref{eq:CFL} is structurally identical to the
forward-Euler stability bound for gradient flows of strongly convex
functionals \cite[Theorem 1.2.4]{StuartHumphries1996}, with the
operator norm of the Hessian replaced by the product
$G\,\norm{\Geps^{-1}}_{2}$. The factor $\norm{\Geps^{-1}}_{2}$
encodes the cost of the projection step; in the unregularised case
it equals the inverse of $\smin^{2}(\Gamma)$, which can be
arbitrarily large in applications.
\end{remark}

\section{Numerical experiments}\label{sec:numerics}

We illustrate the analysis on a three-level bilinear test problem
modelling all-optical generation of a Bell state in two dipole-dipole
coupled atoms.

\subsection{Test problem}

The state space is $\C^{3}$, with basis
$\{|gg\rangle,|s\rangle,|ee\rangle\}$, where
$|s\rangle=(|ge\rangle+|eg\rangle)/\sqrt{2}$ is the symmetric Dicke
state. We set $H_{0}=\diag(-\omega_{0}/2,\,V_{dd},\,\omega_{0}/2)$ and
\begin{equation*}
  \mu=\mu_{d}\begin{pmatrix}0&1&0\\1&0&1\\0&1&0\end{pmatrix},
\end{equation*}
with parameters $\omega_{0}=12578.95\,\mathrm{cm}^{-1}$,
$V_{dd}=12.35\,\mathrm{cm}^{-1}$ and
$\mu_{d}=\sqrt{2}\times 7.61\,\mathrm{D}$ \cite{GuoLuoMaShu2019}.
Atomic units are used throughout. The initial and target states are
$\psi_{0}=(1,0,0)^{\top}=|gg\rangle$ and
$\psi_{f}=(0,1,0)^{\top}=|s\rangle$. The three equality constraints
are
\begin{equation}\label{eq:numerical-constraints}
  h_{1}[E]=\int_{0}^{T}\!E\dd t = 0,\quad
  h_{2}[E]=\int_{0}^{T}\!E^{2}\dd t = C_{2},\quad
  h_{3}[E]=\mu_{d}\!\int_{0}^{T}\!E(t)\cos(\omega_{\mathrm r}t)\dd t = C_{3},
\end{equation}
with reference frequency $\omega_{\mathrm r}=\omega_{0}/2+V_{dd}$.
The time grid contains $N_{t}=4000$ uniform points on
$[-4\tau,4\tau]$ for three pulse durations
$\tau\in\{100,250,400\}\,\mathrm{fs}$, and the gating envelope is
$S(t)=\exp(-t^{2}/(2\tau^{2}))$. The initial control is a
transform-limited Gaussian centred at $\omega_{\mathrm r}$ with area
$\theta_{\mathrm{sg}}=\pi/2$, which fixes the constraint targets
$C_{2}$ and $C_{3}$. Integrals are computed by the composite
trapezoidal rule, the propagator $U_{E}(t)$ is built by exact
eigendecomposition of the Hermitian generator $H_{0}-\mu E(t)$ at
each time slice (equivalent to a zeroth-order Magnus exponential
\cite{BlanesCasasOteoRos2009,IserlesNorsettActa2000}), and the linear
system $\Geps\,x=e_{0}$ is solved by the standard symmetric
positive-definite solver. All computations are performed in double
precision.

\subsection{Baseline behaviour ($\eps=0$)}

Figure~\ref{fig:baseline} reports four diagnostics over the
iterations of the unregularised algorithm at the three pulse
durations: the objective $J$, the relative drift of the fluence
constraint, the absolute pulse area $|h_{1}|$ on logarithmic axes,
and the condition number $\cond(\Gamma(s^{(n)}))$. Two qualitative
features stand out. First, $\cond(\Gamma)$ stays in the band
$10^{9}$--$10^{11}$ throughout every run, so the Gram matrix is
effectively rank-deficient in double precision. Second, the affine
constraints (zero area, fixed reference area) are preserved to
$\sim 10^{-7}$, while the nonlinear fluence constraint drifts by
$20$--$38\%$. This dichotomy is consistent with
Remark~\ref{rem:affine}: only constraints with $E$-dependent
gradients incur a Hessian-driven discretisation drift. Physically,
the fluence drift means that the optimised pulse delivers a larger
energy than the input target $C_{2}$ --- a property of the algorithm,
not of the physics, that would ordinarily go unnoticed without
careful budget tracking.

\begin{figure}[!htbp]
  \centering
  \includegraphics[width=\linewidth]{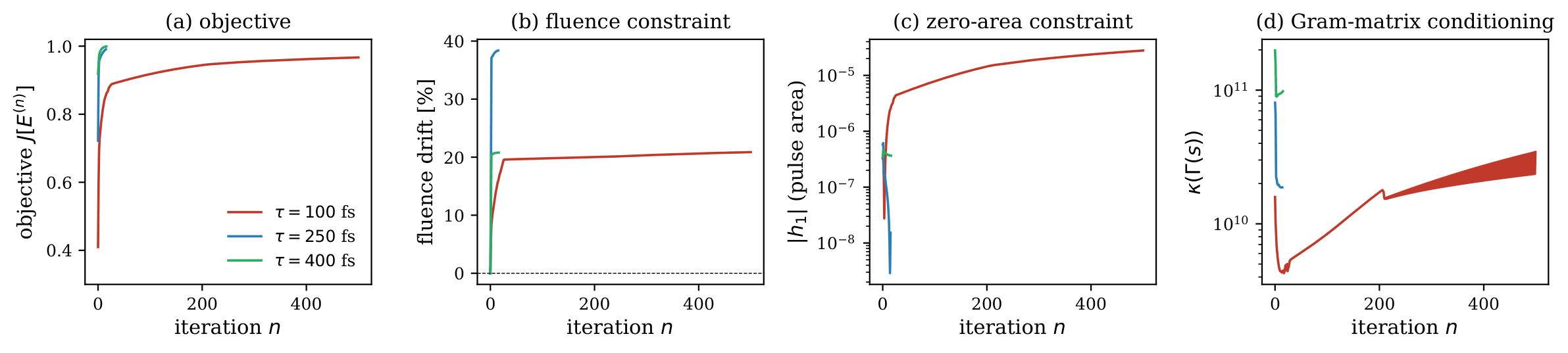}
  \caption{Baseline unregularised flow ($\eps=0$).
    (a) objective $J[E^{(n)}]$;
    (b) relative fluence drift $(h_{2}-h_{2}^{(0)})/h_{2}^{(0)}$;
    (c) $|h_{1}|$ (logarithmic scale);
    (d) Gram-matrix condition number $\cond(\Gamma(s^{(n)}))$.
    The persistence of $\cond(\Gamma)$ in the band
    $10^{9}$--$10^{11}$ across all three pulse durations documents
    the generic ill-conditioning of the unregularised flow.}
  \label{fig:baseline}
\end{figure}

\subsection{Verification of Theorem~\ref{thm:conv}}

Fixing $\tau=250\,\mathrm{fs}$ and running the regularised flow to a
common iteration count, Figure~\ref{fig:thm4} reports the relative
$L^{2}$-distance
$\norm{E_{\eps}(s^{*})-E_{0}(s^{*})}_{L^{2}}/\norm{E_{0}(s^{*})}_{L^{2}}$
against $\eps$. The data lie on a straight line of slope $2$ in
log-log coordinates over eight decades in $\eps$, with a stable
prefactor; the $\mathcal{O}(\eps^{2})$ rate of
Theorem~\ref{thm:conv} is verified empirically. Saturation is
visible for $\eps\gtrsim 10^{-3}$, consistent with the
finite-$\smin$ correction in the constant of
Theorem~\ref{thm:conv}.

\begin{figure}[!htbp]
  \centering
  \includegraphics[width=0.62\linewidth]{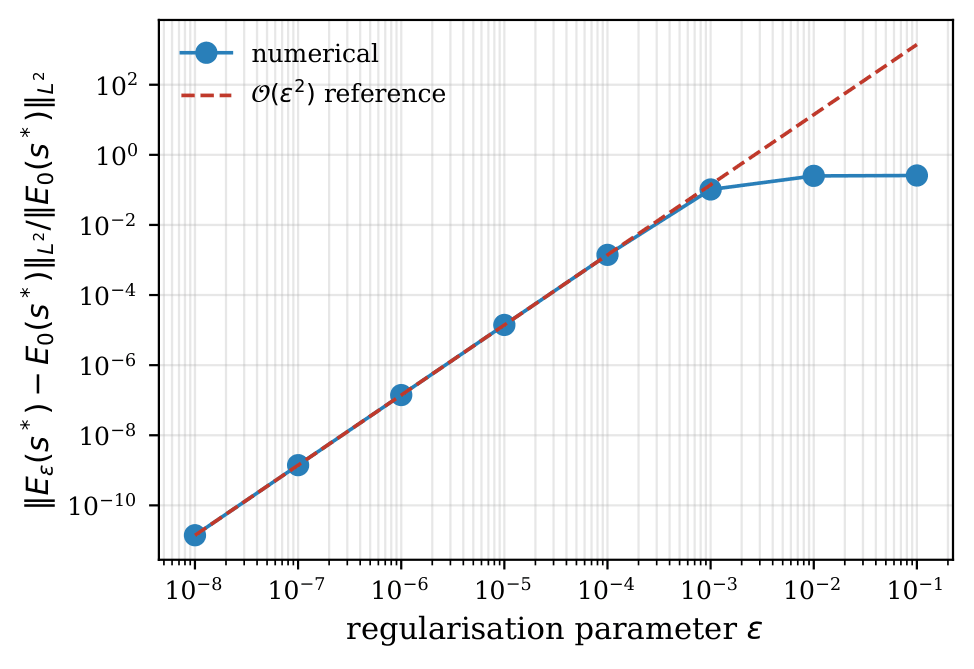}
  \caption{Empirical verification of Theorem~\ref{thm:conv}.
    Markers: $\norm{E_{\eps}(s^{*})-E_{0}(s^{*})}_{L^{2}}
    /\norm{E_{0}(s^{*})}_{L^{2}}$ against $\eps$. Dashed line:
    $\mathcal{O}(\eps^{2})$ reference. The $\eps^{2}$ rate is
    reproduced over eight decades in $\eps$.}
  \label{fig:thm4}
\end{figure}

\subsection{Conditioning--drift trade-off}

Figure~\ref{fig:cond-drift} plots the maximum condition number
$\max_{s}\cond(\Geps(s))$ along the trajectory and the relative
fluence drift at the final iterate, both as functions of $\eps$. The
condition number tracks the bound~\eqref{eq:cond-formula}: it is
$\sim 10^{11}$ for $\eps\lesssim 10^{-4}$ (the small-$\eps$ regime
of Remark~\ref{rem:cond-interp}), decreases as $\eps^{-2}$ for
$\eps\in[10^{-4},10^{-1}]$, and approaches $1$ for larger $\eps$.
The fluence drift, dominated by discretisation error in the
small-$\eps$ regime, decreases to about $3\%$ as $\eps$ grows to
$10^{-2}$.

\begin{figure}[!htbp]
  \centering
  \includegraphics[width=\linewidth]{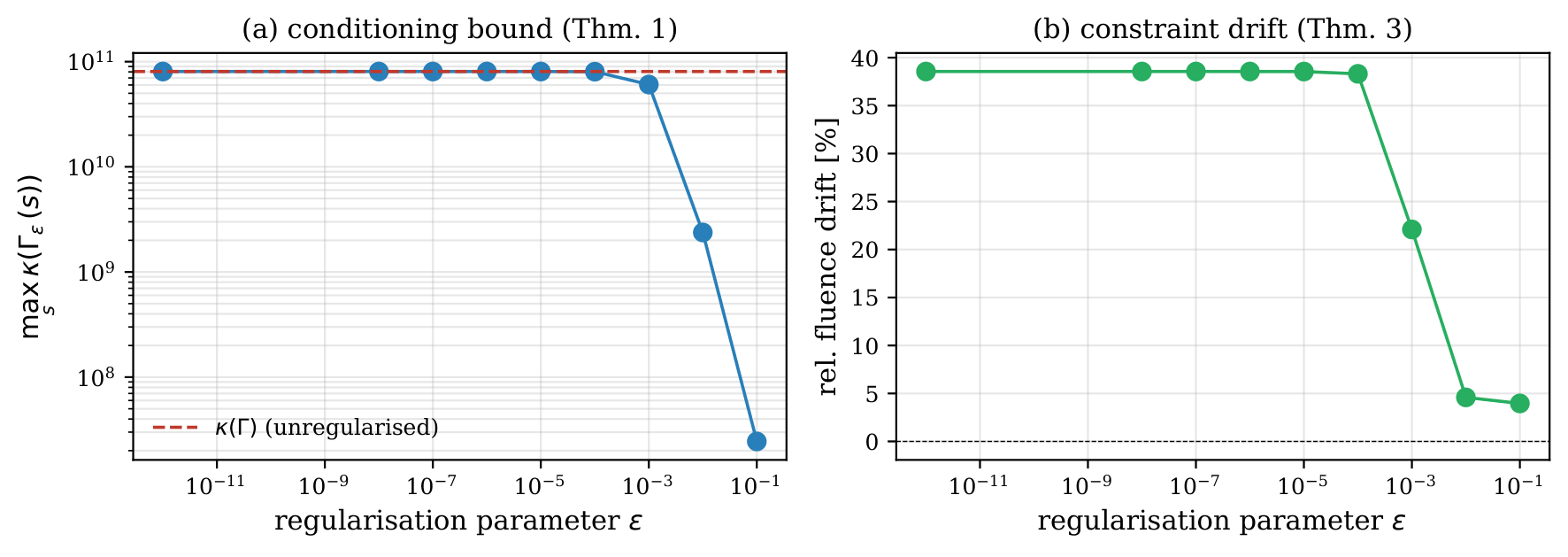}
  \caption{(a) Maximum condition number $\max_{s}\cond(\Geps(s))$
    along the regularised flow versus $\eps$. The plateau at small
    $\eps$ matches $\cond(\Gamma)$ (dashed line). (b) Relative
    fluence drift at the final iterate. The decrease at moderate
    $\eps$ reflects the CFL margin of Theorem~\ref{thm:euler}.}
  \label{fig:cond-drift}
\end{figure}

\subsection{Practical impact of regularisation}

We test the practical impact of regularisation in the
aggressive-step regime. For initial step sizes
$\Delta s\in\{10^{-6},5\cdot 10^{-6},10^{-5},5\cdot 10^{-5},10^{-4}\}$
and $\eps\in\{0,10^{-4},10^{-3},10^{-2}\}$, we record (i) the
iterations to reach $J=0.99$, (ii) the number of step-halving
rejections, and (iii) the accumulated fluence drift
(Figure~\ref{fig:payoff}). At $\Delta s=10^{-6}$, the regularised
algorithm with $\eps=10^{-2}$ converges in six iterations, against
fifteen for the unregularised run. Regularisation with
$\eps=10^{-2}$ eliminates step rejections at $\Delta s=10^{-6}$, in
agreement with Theorem~\ref{thm:euler}: the CFL bound is easier to
satisfy when $\norm{\Geps^{-1}}$ is smaller. The accumulated fluence
drift falls from $38\%$ to $2.8\%$ at the same final fidelity, an
order-of-magnitude reduction. For $\Delta s\gtrsim 5\cdot 10^{-5}$
the bound~\eqref{eq:CFL} is violated for every $\eps$ tested and the
discretisation breaks down, in agreement with
Theorem~\ref{thm:euler}.

\begin{figure}[!htbp]
  \centering
  \includegraphics[width=\linewidth]{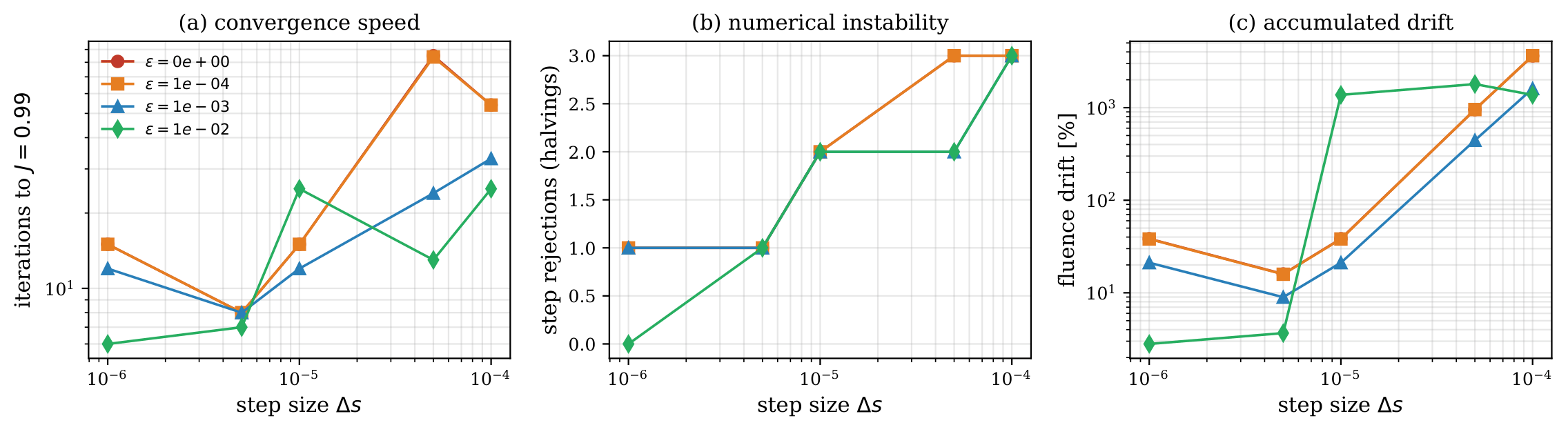}
  \caption{Practical impact of regularisation at
    $\tau=250\,\mathrm{fs}$.
    (a) Iterations to reach $J=0.99$ versus initial step size
        $\Delta s$.
    (b) Total step-halving rejections.
    (c) Relative fluence drift at termination.
    Regularisation with $\eps\in\{10^{-3},10^{-2}\}$ outperforms the
    unregularised baseline on all three metrics; aggressive step
    sizes violate the CFL bound for every $\eps$.}
  \label{fig:payoff}
\end{figure}

\subsection{Final fields}

For completeness, Figure~\ref{fig:fields} compares the initial
control to the final controls produced by $\eps=0$, $\eps=10^{-3}$,
and $\eps=10^{-2}$. All three runs reach $J\ge 0.9999$, with target
fidelity for the Bell state $|s\rangle$ exceeding $99.99\%$. The
$L^{2}$-distances between the regularised and unregularised final
controls are consistent with Theorem~\ref{thm:conv} and
Figure~\ref{fig:thm4}.

\begin{figure}[!htbp]
  \centering
  \includegraphics[width=0.78\linewidth]{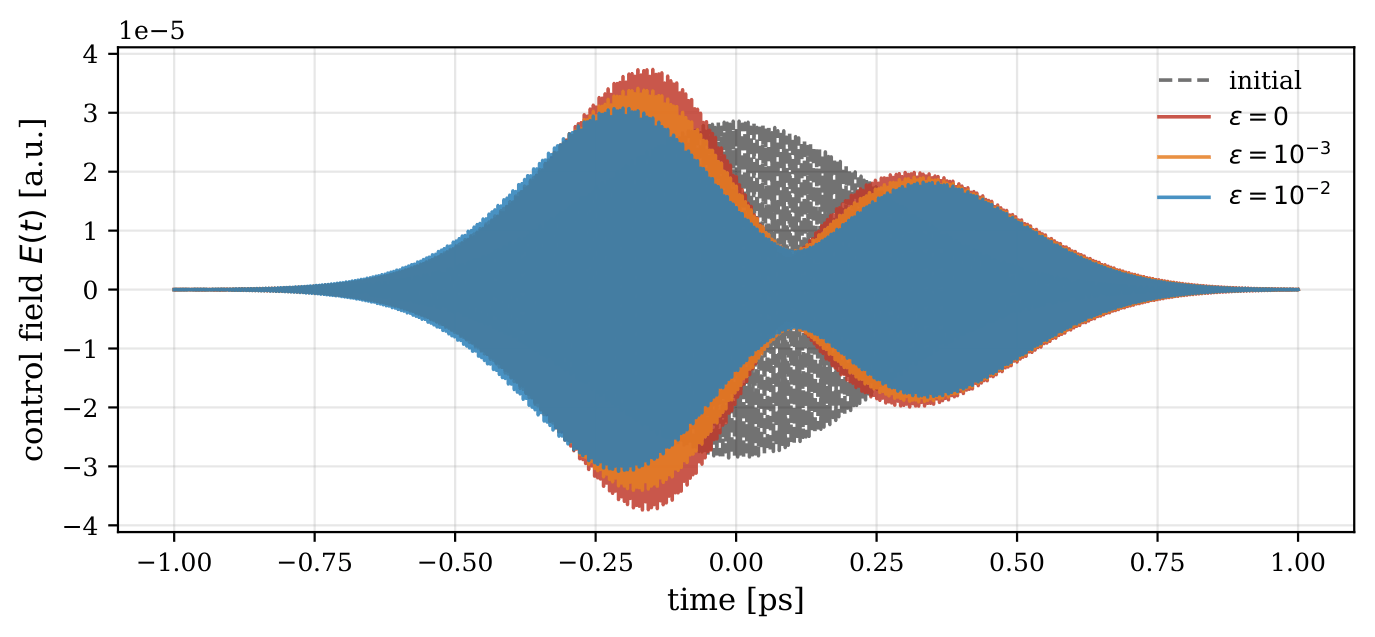}
  \caption{Initial control (dashed) and final controls obtained
    with $\eps=0$, $\eps=10^{-3}$, $\eps=10^{-2}$. All three
    controls produce $J\ge 0.9999$; the $L^{2}$-distances among
    them match the predictions of Theorem~\ref{thm:conv}.}
  \label{fig:fields}
\end{figure}

\section{Discussion}\label{sec:discussion}

The numerical experiments admit a transparent physical reading. The
condition number of the moving Gram matrix measures the degree to
which the constraint gradients $c_{1},c_{2},c_{3}$ become
approximately collinear when restricted to the support of the gating
envelope $S$. In the benchmark of Section~\ref{sec:numerics} the
dominant near-collinearity arises from the interplay between the
zero-area constraint $c_{1}\equiv 1$ and the resonant-area
constraint $c_{3}=\mu_{d}\cos(\omega_{\mathrm r}t)$: when integrated
against $S$, the oscillatory part of $c_{3}$ averages to a small
residual and $c_{1}$ becomes nearly proportional to a slowly varying
mean of $c_{3}$ over the pulse window. This feature is generic for
constraint sets that mix a DC-like component with an oscillatory
component at a finite carrier frequency, and it explains why the
ill-conditioning we document is not a peculiarity of the chosen
parameters but a structural property of three-constraint pulse-design
problems.

Beyond stabilising the discretisation, Tikhonov regularisation has a
direct interpretation as a smooth interpolation between two natural
algorithms. At $\eps=0$ the flow projects the gradient $c_{0}$ onto
the orthogonal complement of the constraint span in the $S$-weighted
inner product; in the limit $\eps\to\infty$ it reduces to ordinary
gradient ascent on $J$, ignoring the constraints altogether. The
intermediate regime $\eps\sim\smin$ realises a partial projection in
which low-energy modes of $\Gamma$ --- combinations of constraint
gradients poorly resolved by the trapezoidal quadrature --- are
damped rather than enforced exactly. The damping introduces a
controllable $\mathcal{O}(\eps^{2})$ constraint drift but eliminates
the discrete instability that would otherwise be triggered by
attempted exact enforcement of an effectively underdetermined linear
system. The trade-off quantified by Theorems~\ref{thm:drift}
and~\ref{thm:conv} is therefore not a workaround: it is a principled
rebalancing of the cost of constraint preservation against the cost
of numerical stability.

For applications in which the unregularised algorithm of
\cite{ShuHoRabitz2016} would be the method of choice ---
all-optical Bell-state preparation \cite{GuoLuoMaShu2019},
phase-locked state transfer \cite{GuoDongShu2018}, and orientation
control with zero-area fields \cite{SugnyConstraints2014} --- our
results suggest two practical recommendations. First, the
unregularised flow should be regarded as a useful continuous-time
abstraction rather than a discrete recipe; any implementation
inevitably involves an ad hoc step-size policy whose effect is hard
to reason about. Second, choosing $\eps$ in the range
$10^{-3}$--$10^{-2}$, suitably non-dimensionalised by the spectral
scale of the problem, gives a transparent CFL bound, eliminates the
overhead of objective-decrease probes, and reduces accumulated
constraint drift by an order of magnitude at unchanged objective
fidelity.

The framework is independent of the bilinear form of the state
equation~\eqref{eq:tdse}: any control problem in which the gradient
$c_{0}$ and the constraint gradients $c_{m}$ are well-defined Hilbert
elements admits the same construction, and the conditioning analysis
is purely linear-algebraic. Application of the regularised scheme to
nonlinear-state, multi-control, and partial-differential-equation-constrained
problems is a natural next step. From the physics side, the recent
rapid progress in neutral-atom quantum technology
\cite{Saffman2016,BrowaeysLahaye2020,Evered2023,Bluvstein2024}, in
which Rydberg-mediated entangling gates with fidelities approaching
$99.5\%$ are now routinely demonstrated, gives the constrained
pulse-design problem renewed practical urgency, and the benchmark of
Section~\ref{sec:numerics} is a stepping-stone toward a fully
analytical treatment of pulse-design problems on those platforms.

\section{Conclusions and outlook}\label{sec:conclusion}

We have analysed the Tikhonov-regularised projected gradient flow
for equality-constrained bilinear optimal control. The analysis
yields five results: an exact spectral identity for the regularised
Gram matrix, leading to an explicit conditioning formula
(Theorem~\ref{thm:cond}); preservation of objective monotonicity for
every $\eps\ge 0$ (Theorem~\ref{thm:mono}); an exact constraint-drift
identity, yielding $\mathcal{O}(\eps^{2})$ drift bounds with
computable prefactors (Theorem~\ref{thm:drift}); $L^{2}$-convergence
of the regularised to the unregularised flow at rate
$\mathcal{O}(\eps^{2})$ under uniform invertibility of $\Gamma$
(Theorem~\ref{thm:conv}); and a discrete CFL-type criterion
$\Delta s\,G\,\norm{\Geps^{-1}}\le\alpha<2$ guaranteeing discrete
objective monotonicity up to $\mathcal{O}(\Delta s^{2})$ local
truncation error (Theorem~\ref{thm:euler}). On the three-level
bilinear benchmark relevant to all-optical Bell-state preparation,
all five predictions are corroborated; the $\mathcal{O}(\eps^{2})$
rate is reproduced over eight decades in $\eps$, and moderate
regularisation eliminates the discrete step rejections observed in
unregularised implementations while reducing the
discretisation-induced constraint drift by an order of magnitude.

Several directions remain open. The uniform-invertibility hypothesis
in Theorem~\ref{thm:conv} is restrictive when $\Gamma$ becomes
near-singular along part of the trajectory; replacing it by an
$s$-dependent regularisation $\eps(s)$ in the spirit of the
Tikhonov--Morozov discrepancy principle
\cite{EnglHankeNeubauer1996,Calvetti2003} should permit a local
analysis. The recent convergence theory for stochastic gradient
methods with decaying regularisation \cite{ControllingTheFlow2025}
suggests template ideas. The CFL bound~\eqref{eq:CFL} uses the
global operator norm $\norm{\Geps^{-1}}$, which is conservative; a
sharper, eigenvector-adapted choice of $\Delta s$ should be possible.
An analysis of higher-order Runge--Kutta integrators
of~\eqref{eq:floweps} would partially compensate the second-order
Taylor remainder in Theorem~\ref{thm:euler}, potentially weakening
the CFL constant. Extension to inequality constraints via active-set
or slack-variable lifts \cite{NocedalWright2006,UlbrichSemismooth2011}
should transfer with minor modification. Application of the
regularised scheme to multi-control problems --- two-frequency or
polarisation-resolved pulse design for Rydberg-mediated gates ---
would extend the practical reach of the analysis to the platforms of
greatest current experimental interest.

\section*{Acknowledgements}
The author thanks the referees and editors of the journal for their
attention. The benchmark problem of Section~\ref{sec:numerics} was
the subject of earlier work \cite{GuoLuoMaShu2019}; the
numerical-analysis viewpoint, the theorems and the implementation
are the author's own.

\section*{Data availability}
The numerical experiments reported in Section~\ref{sec:numerics} are
fully described by Algorithm~\ref{alg:reg-dmorph}, the parameters
listed in~\eqref{eq:numerical-constraints} and the surrounding text,
and the following implementation choices: composite trapezoidal
quadrature on a uniform grid of $N_{t}=4000$ points; exact
eigendecomposition of $H_{0}-\mu E(t)$ at each time slice for the
propagator update; and solution of the linear system
$\Geps\,x=e_{0}$ by the standard symmetric positive-definite solver.
With these specifications, all reported numbers and figures are
reproducible to plotting accuracy.

\bibliographystyle{unsrt}
\bibliography{references}

\appendix

\section{Auxiliary lemmas on the Gram matrix}\label{app:gram}

We collect two elementary lemmas used in the body of the paper,
included for completeness.

\begin{lemma}[Diagonal-entry bounds for $\Geps^{-1}$]
\label{lem:diag-bound}
For every symmetric positive semidefinite
$\Gamma\in\R^{(M+1)\times(M+1)}$ and every $\eps>0$, the diagonal
entries of $\Geps^{-1}=(\Gamma+\eps^{2}I)^{-1}$ satisfy
\begin{equation}\label{eq:diag-bound}
  \frac{1}{\Gamma_{kk}+\eps^{2}}
  \le\bigl[\Geps^{-1}\bigr]_{kk}
  \le\frac{1}{\smin(\Gamma)^{2}+\eps^{2}}
  \qquad\text{for every }k=0,1,\dots,M.
\end{equation}
\end{lemma}

\begin{proof}
The upper bound is the operator-norm bound applied to the diagonal
entry. For the lower bound, observe that
$[\Geps]_{kk}=\Gamma_{kk}+\eps^{2}$ and apply the inequality
$[A^{-1}]_{kk}\ge 1/A_{kk}$, valid for any symmetric positive-definite
$A$ \cite[Section 7.7]{HornJohnson2013}.
\end{proof}

\begin{lemma}[Cauchy--Schwarz on the off-diagonal entry]
\label{lem:cs}
For every symmetric positive-definite
$A\in\R^{(M+1)\times(M+1)}$ and distinct indices $i,j$,
\begin{equation}\label{eq:cs}
  \bigl|[A^{-1}]_{ij}\bigr|
  \le\bigl[A^{-1}\bigr]_{ii}^{1/2}\bigl[A^{-1}\bigr]_{jj}^{1/2}.
\end{equation}
\end{lemma}

\begin{proof}
Write $[A^{-1}]_{ij}=e_{i}^{\top}A^{-1}e_{j}
=\inner{A^{-1/2}e_{i}}{A^{-1/2}e_{j}}_{\R^{M+1}}$ and apply
Cauchy--Schwarz.
\end{proof}

Together, Lemmas~\ref{lem:diag-bound} and~\ref{lem:cs} yield the
bound $|[\Geps^{-1}]_{m0}|\le 1/(\smin^{2}+\eps^{2})$ used in the
proof of Theorem~\ref{thm:drift}.

\section{Notational conventions}\label{app:notation}

For the convenience of the reader, we collect the principal
notational conventions used in the paper.
\begin{center}
\begin{tabular}{ll}
\toprule
Symbol & Meaning \\ \midrule
$\Hil=L^{2}(0,T;\R)$ & Hilbert space of admissible controls \\
$\inner{\cdot}{\cdot}$, $\norm{\cdot}$ & inner product and norm on $\Hil$ \\
$\norm{\cdot}_{2}$ & Euclidean / spectral norm \\
$\norm{\cdot}_{\infty}$ & $L^{\infty}$-norm on $[0,T]$ \\
$E$, $E_{\eps}$ & control field; regularised control field \\
$\psi(t)=U_{E}(t)\psi_{0}$ & state at time $t$ \\
$\lambda(t)$ & costate, equation~\eqref{eq:lambda-def} \\
$J[E]=|\inner{\psi_{f}}{U_{E}(T)\psi_{0}}|^{2}$ & terminal fidelity \\
$h_{m}[E]=C_{m}$, $m=1,\dots,M$ & equality constraints \\
$c_{0}=\delta J/\delta E$ & objective gradient \\
$c_{m}=\delta h_{m}/\delta E$ & constraint gradients \\
$\bdf=(c_{0},c_{1},\dots,c_{M})^{\top}$ & gradient column vector \\
$\Mc=\{E\in\Hil:h_{m}[E]=C_{m}\}$ & constraint manifold \\
$S\in C^{2}([0,T])$ & gating envelope, $\norm{S}_{\infty}=1$ \\
$\Gamma$, $\Geps=\Gamma+\eps^{2}I$ & Gram matrix and Tikhonov regularisation \\
$\sigma_{k}$, $\smin$, $\smax$ & singular values of $C_{S}$ \\
$\cond$, $\spec$ & condition number, spectrum \\
$g_{0}(s)=\Gamma_{00}(s)$, $g_{0}^{(\eps)}(s)$ & scalars in the flow \\
$s\in[0,s^{*}]$ & morphing variable \\
$\Delta s$ & discrete step size in $s$ \\
$\eps\ge 0$ & Tikhonov regularisation parameter \\
$\alpha\in(0,2)$ & CFL safety factor \\
$G$ & bound on the second variation of $J$ \\
$\tau$ & pulse duration in the numerical example \\
$T$ & terminal time in the time grid \\
\bottomrule
\end{tabular}
\end{center}

\end{document}